\documentclass{llncs} 
\usepackage{amsmath}
\usepackage{amssymb}
\usepackage[english,vlined,ruled]{algorithm2e}
\usepackage{pstricks,pst-node}

\begin{document}
\title{ON THE COMPLEXITY OF THE MULTIPLE STACK TSP, $k$STSP}
\author{Sophie Toulouse and Roberto Wolfler Calvo}
\institute{LIPN (UMR CNRS 7030) - Institut Galil\'ee, Universit\'e Paris 13\\99 av. Jean-Baptiste Cl\'ement, 93430 Villetaneuse, France.\\
			\email{sophie.toulouse@lipn.univ-paris13.fr, wolfler@lipn.univ-paris13.fr}
}

\maketitle
\begin{abstract}Given a universal constant $k$, the multiple Stack Travelling Salesman Problem ($k$STSP in short) consists in finding a pickup tour $T^1$ and a delivery tour $T^2$ of $n$ items on two distinct graphs. The pickup tour successively stores the items at the top of $k$ containers, whereas the delivery tour successively picks the items at the current top of the containers: thus, the couple of tours are subject to LIFO ({\em ``Last In First Out''}) constraints. This paper aims at finely characterizing the complexity of $k$STSP in regards to the complexity of TSP. First, we exhibit tractable sub-problems: on the one hand, given two tours $T^1$ and $T^2$, deciding whether $T^1$ and $T^2$ are compatible can be done within polynomial time; on the other hand, given an ordering of the $n$ items into the $k$ containers, the optimal tours can also be computed within polynomial time. Note that, to the best of our knowledge, the only family of combinatorial precedence constraints for which constrained TSP has been proven to be in $\mathbf{P}$ is the one of PQ-trees, \cite{BDW96}. Finally, in a more prospective way and having in mind the design of approximation algorithms, we study the relationship between optimal value of different TSP problems and the optimal value of $k$STSP.
\end{abstract}

\section{Introduction}\label{sec-intro}
\subsection{The problem specification}\label{sec-intro-pb}
Assume that a postal operator has to pick up $n$ items in some city $1$, and then to deliver the same items in some city $2$. Such a situation can be modelized by means of two TSP instances $I^1=(G^1,d^1)$ and $I^2=(G^2,d^2)$, where the two graphs $G^1=(V^1,E^1)$ and $G^2=(V^2,E^2)$ have the same order $n+1$ (vertex $0$ represents the depot, whereas vertices $[n]$ represent the location where the items have to be picked up or delivered), and the distance functions $d^1:E^1\rightarrow\mathbb{N},d^2:E^2\rightarrow\mathbb{N}$ associate integer values to the edges of $G^1,G^2$. The two TSP instances $I^1,I^2$ thus represent the search of an optimal pickup tour in city $1$, and the search of an optimal delivery tour in city $2$, respectively. If no constraint occurs between the two tours, then the problem is equivalent to the resolution of two independent TSP. In $k$STSP, one assumes that the tours are subject to LIFO contraints, namely: the pickup tour stacks the items into $k$ containers, so that the delivery tour must deliver {\em at first} the items that have been stored {\em at last} by the pickup tour. Hence, a solution of $k$STSP consists of a couple of tours $(T^1,T^2)$, together with a stacking order $\mathcal{P}$ on the $k$ containers that is compatible with both the pickup and the delivery tours. Here, a stacking order is defined as a set $\{P^1,\ldots,P^k\}$ of $k$ $q^\ell$-uples $P^\ell=(v^\ell_{1},\ldots,v^\ell_{q^\ell})$ that partitions $[n]$, where vertices $v^\ell_1$ and $v^\ell_{q^\ell}$ respectively represent the bottom and the top of the $\ell$th stack. A feasible solution $(T^1,T^2,\mathcal{P})$ is optimal if it is of optimal distance, where the distance is given by the sum $d^1(T^1)+d^2(T^2)$ of the distances of the two tours. For sake of simplicity, we will always consider that $G^1$ and $G^1$ are the complete directed graph ${\vec{K}}_{n+1}$ on $V=\{0\}\cup[n]$. In the case of symetric distance functions, one just has to consider $d^\alpha(u,v)=d^\alpha(v,u)$ for $u,v\in V$ and $\alpha\in\{1,2\}$; moreover, one could recognize unexisting arcs by associating to each couple of vertices $u,v\in V$ such that $(u,v)\notin E^\alpha$, e.g., the distance $d^{\alpha}(u,v)=d^{\alpha}_{max}+1$, where $d^{\alpha}_{max}=\max\{d^\alpha(e)\ |\ e\in E^\alpha\}$.

\subsection{Related work \& outline}\label{sec-intro-lit}
By contrast to other TSP problems, only a few literature exists on this problem (see, for example, \cite{FOT08,PM08} for some heuristic approaches), and none (to the best of our knowledge) about its complexity. Anyway, the problem we address naturally is $\mathbf{NP-hard}$, from TSP. 

Nevertheless, one could wonder about what combinatorial structure of $k$STSP impacts more on its complexity: the stacks (and the LIFO constraints they induce on the tours), or the permutations themselves? The answer is not clear and this paper shows why. First of all we prove that, given two tours $T^1$ and $T^2$, deciding whether $T^1$ and $T^2$ are compatible or not is tractable, since the decision reduces to $k$-coloring in comparability graphs. Moreover, given an ordering of the $n$ items into the $k$ stacks, the optimal tours can also be computed within polynomial time, by means of dynamic programming. Another interesting question concerns the relative complexity of $k$STSP in regards to TSP: how much its trickier combinatorial structure makes $k$STSP harder to solve (exactly as well as approximatively) than TSP? Although we do not provide formal answers to this latter question, we give some intuition that $k$STSP is globally harder than general TSP to optimize: it is obvious that efficient algorithms for $k$STSP can be derived in order to solve TSP; by contrast, we establish that tours of good quality for TSP may lead to arbitrary low quality solutions for STSP. 

The paper is organized as follows: we first expose in section~\ref{sec-pte} some notations and properties that will be useful for the next sections; section~\ref{sec-cpx} exhibits two tractable sub-problems (one of decision when the tours are given, \ref{sec-cpx-tour}, one of optimization when the stacks are given, \ref{sec-cpx-stack}); section~\ref{sec-apx} compares the behaviour of solution values in STSP instances and some related TSP instances, bringing to the fore that good resolution of TSP may not lead to good resolution of STSP; finally, section~\ref{sec-conc} concludes with some perspectives. 

\section{Preliminaries}\label{sec-pte}
Three strict orders $<^1$, $<^2$, $<^3$ on $[n]$ are associated, respectively, to the pickup tour $T^1=(0, u^1_1, \ldots u^1_n, 0)$, to the delivery tour $T^2=(0, u^2_1, \ldots u^2_n, 0)$ and to a stacking order $\mathcal{P}=\{P^1,\ldots,P^k\}$. The two orders $<^1$, $<^2$ are complete whereas $<^3$ is partial. It means that $\forall a\neq b\in[n]$ we can write:
$$\begin{array}{|lccl|}\hline
	<^1: 	&a<^1 b 														&\Leftrightarrow 	&T^1\mbox{ picks up }a\mbox{ before }b\\
	<^2: 	&a<^2 b 														&\Leftrightarrow 	&T^2\mbox{ delivers }b\mbox{ before }a\\\hline
	<^3: 	&a<^3 b															&\Leftrightarrow	&\exists \ell\in[1,k]\ /\ a,b\in P^\ell,\ a\mbox{ is stacked before }b\mbox{ in }P^\ell\\
				&\neg (a<^3 b) \wedge \neg (a<^3 b)	&\Leftrightarrow 	&a,b\mbox{ are stacked into two distinct stacks}\\\hline
\end{array}$$

\begin{lemma}\label{lem-order}
A solution $(T^1, T^2,\mathcal{P})$ is feasible iff the three orders $<^1$, $<^2$, $<^3$ it induces on $[n]$ satisfy the following conditions:
\begin{eqnarray}	
\forall a\neq b\in[n],\ a<^1 b	\Rightarrow	\neg (a>^3 b) \label{uno}\\
\forall a\neq b\in[n],\ a<^2 b	\Rightarrow \neg (a>^3 b) \label{due}
\end{eqnarray}
\end{lemma}

\begin{proof}
The necessary condition is obvious. For the sufficient condition, let consider a pickup tour $T^1$ and a stacking order $\mathcal{P}$ (the argument is rather similar for the delivery tour). For any $i\in[n]$, $T^1$ has to pick up the item $u^1_i$, that is of index $j$ in some stack $P^\ell$; this is possible iff the previous item that $T^1$ has picked up in $P^\ell$ is the item $u^\ell_{j-1}$, or there is no such index and $j=1$, what is always true if $T^1$ and $\mathcal{P}$ verify condition (\ref{uno}).
\qed\end{proof}

\section{Complexity classes and properties}\label{sec-cpx}

\subsection{Global complexity}\label{sec-cpx-glob}
The problem obviously is $\mathbf{NP-hard}$, for arbitrary instances $(G^1,d^1;G^2,d^2)$ of $k$STSP (where by {\em ``arbitrary''}, we mean that we do not make any asumption, neither on the graph completeness, nor on the symetry of the distance functions). When the distance functions $d^1$ and $d^2$ are the same, up to the arc direction (that is, $d^1(a,b)=d^2(b,a)$ for all $a,b\in\{0\}\cup[n]$), it is equivalent to the regular TSP (consider on the one hand that $T^1=(0, u_1,\ldots, u_n, 0)$ is an optimal pickup tour iff $T^2=(0, u_n, \ldots, u_1, 0)$ is an optimal delivery tour, on the other hand that every stacking order that is feasible for $T^1$ also is feasible for $T^2$). Second, for a given triple $(T^1,T^2,\mathcal{P})$ where $T^1$, $T^2$ are two tours on $\{0\}\cup[n]$ and $\mathcal{P}$ is a stacking order of $[n]$ using $k$ stacks, checking whether $(T^1,T^2,\mathcal{P})$ is feasible or not can be done within linear time (quite immediate from Lemma~\ref{lem-order}). 

\subsection{Deciding feasibility for a couple of tours}\label{sec-cpx-tour}
Let us denote by $G^{\neq}=(V^{\neq},E^{\neq})$ the graph induced by the set of pairs $\{a,b\}$ such that the two orders $<^1$ and $<^2$ are discordant:
	$$E^{\neq}= \{\ \{a,b\}\    |\ a\neq b\in [n],\ a<^1 b\ \wedge\ a>^2 b\ \},\ V^{\neq}=\bigcup_{\{a,b\}\in E^{\neq}}\{a,b\}$$

\begin{lemma}\label{lem-coloration}
	Given two tours $T^1$, $T^2$, a compatible stacking order $\mathcal{P}$ exists iff $\chi(G^{\neq})\leq k$, where $\chi(G^{\neq})$ denotes the chromatic number on $G^{\neq}$.
\end{lemma}

\begin{proof}
For the necessary condition, consider a feasible solution $(T^1, T^2,\mathcal{P})$ and two items $a\neq b$ such that $\{a,b\}\in E^{\neq}$, iff $a<^1 \wedge a>^2 b$, or $a>^1 b \wedge a<^2 b$. In both cases, we know from Lemma~\ref{lem-order} that $\neg(a>^3b) \wedge \neg(a>^3b)$; thus, the $k$ stacks in $\mathcal{P}$ correspond to $k$ independent sets in $G^{\neq}$. For the sufficient condition, we build from a $k$-coloring on $V^{\neq}$ a stacking order $\mathcal{P}$ that fulfills, together with $T^1$ and $T^2$, conditions (\ref{uno}) and (\ref{due}) of Lemma~\ref{lem-order}. We first fill each stack $P^\ell$ with the items of the $\ell$th color set $\{v^\ell_1,\ldots,v^\ell_{q^\ell}\}$, by considering on $P^\ell$ the order induced by the relation $<^{1,2}$ defined as: $<^{1,2}=<^1\wedge<^2$ (the two orders do coincide on each color $\ell$). It remains to insert into the stacks the items from $[n]\backslash V^{\neq}$. The orders $<^1$ and $<^2$ also coincide on $([n]\backslash V^{\neq})\times [n]$; we can therefore write $[n]\backslash V^{\neq}=(v_1,\ldots,v_r)$ with $v_1<^{1,2}\ldots<^{1,2} v_r$. For index $i$ from $1$ to $r$, we insert $v_i$ in position $j(v_i)+1$ in $P^1$ iff $j(v_i)$ is the current maximum index $j$ such that $v^1_j<^{1,2}v_i$, if such an index exists; otherwise, $j(v_i)=0$ (in any case, indices in $P^1$ are updated after each insertion). We finally obtain a partition of $[n]$ within a set of $k$ stacks such that in every stack, the elements are ordered with respect to $<^{1,2}$, what fulfills conditions (\ref{uno}) and (\ref{due}).
\qed\end{proof}

Graph coloring problems (in general, but also $k$-coloring for a universal constant $k\geq 3$) are known to be $\mathbf{NP-c}$ (see, {\em e.g.}, \cite{GJ79}); nevertheless, it turns out that $G^{\neq}$ belongs to the class of perfect graphs, for which determining $\chi(G)$ is in $\mathbf{P}$, \cite{GLS81}. Hence, the considered decision problem is tractable, for any $k$ (and this even if $k$ is not any longer considered as a universal constant, but as being part of the input).

\begin{theorem}\label{thm-coloration}
The STSP sub-problem that consists, given a couple $(T^1,T^2)$, in deciding whether there exists or not a compatible stacking order, is in $\mathbf{P}$.
\end{theorem}

\begin{proof}
$E^{\neq}$ represents the pairs $\{a,b\}$ of $[n]$ such that $(a<^1 b\wedge a>^2 b)$ or $(a>^1 b\wedge a<^2 b)$, where we recall that $<^1$ and $<^2$ both totally order $[n]$. Therefore, $G^{\neq}$ is a comparability graph. Indeed, consider the set of arcs $F^{\neq}=\{(a,b)\in [n]\times[n]\ |\ a<^1 b\wedge a>^2 b\}$: $(i)$~for all $a\neq b \in V^{\neq}$, $(a,b)\in F^{\neq}\vee(b,a)\in F^{\neq}$ iff $\{a,b\}\in E^{\neq}$; $(ii)$~for all $a\neq b \in V^{\neq}$, $(a,b)\in F^{\neq}\Rightarrow (b,a)\notin F^{\neq}$; $(iii)$~for all distinct $a,b,c \in V^{\neq}$, $(a,b)\in F^{\neq}$ and $(b,c)\in F^{\neq}$ iff $a<^1 b<^1 c \wedge a>^2 b>^2 c$ and thus, $(a,c)\in F^{\neq}$. Then, $F^{\neq}$ defines a transitive orientation of the edge set $E^{\neq}$, and $G^{\neq}$ is a comparability graph. By the way, note that a comparability graph $G=(V,E)$ may represent the conflict graph of some couple or orders $(<^1,<^2)$ iff its complementary graph $\overline{G}=(V,\overline{E})$ also is a comparability graph (representing $<^1\wedge<^2$).
\begin{algorithm}\label{algo-stack-from-tour}
\caption{STACKING FROM $T^1$ AND $T^2$}
\KwIn{$T^1$ a pickup tour, $T^2$ a delevery tour, $k$ the number of stacks.}
\KwOut{A compatible stacking order $\mathcal{P}$ iff $(T^1,T^2)$ is feasible.}
\BlankLine

{\bf for} $\ell=1$ {\bf to} $k$ {\bf do } $P^\ell\longleftarrow\emptyset$\;
\BlankLine
build $G^{\neq}=(V^{\neq},E^{\neq})$ from $T^1$, $T^2$\; 
\BlankLine

\tcp{Coloring stage}
{\bf for each} $v\in [n]\backslash V^{\neq}$ {\bf do} $\mathcal{C}(v) \longleftarrow 1$\;
\BlankLine
compute a minimum coloring $\left(C:V^{\neq}\rightarrow [\chi(G^{\neq})],\ v\mapsto C(v)\right)$ on $G^{\neq}$\;
{\bf if} $\chi(G^{\neq}) > k$ {\bf then return} NO\;
\BlankLine

\tcp{Stacking stage (done according to $<^1$)}
{\bf for} $i=1$ {\bf to} $n$ {\bf do} $\{$ $\ell\longleftarrow C(u^1_i)$; stack $u^1_i$ into $P^{\ell}$ $\}$\;
\BlankLine

\Return $\mathcal{P}=\{P^1,\ldots,P^k\}$\;
\end{algorithm}
Algorithm~\ref{algo-stack-from-tour} is a polynomial time procedure that, given a couple of tours $(T^1,T^2)$, responses NO if this couple is unfeasible, returns a compatible stacking order $\mathcal{P}$ otherwise.
\qed\end{proof}

\subsection{Optimizing the tours when the stacks are given}\label{sec-cpx-stack}
In this section, we prove that it is a tractable problem to compute the optimal tours, when the stacks are fixed. 
\begin{theorem}\label{thm-TSP}
For a given stacking order $\mathcal{P}$, one can find an optimal pickup tour and an optimal delivery tour within polynomial time (but exponential in $k$).
\end{theorem}

\begin{proof} We only present the argument for the pickup tour (the proof being rather similar for the delivery tour). Let $\mathcal{P}=P^1,\ldots,P^k$ where $P^\ell=(v^\ell_1,\dots,v^\ell_{q^\ell})$ for any $\ell$ be a stacking order. A pickup tour that is compatible with $\mathcal{P}$ starts by picking up the items which must be placed at the bottom of each stack, until all the stacks have been completely read. Hence, we will consider the space of states $\mathcal{S}=\times_{\ell=1}^k[0,q^\ell]$, where a state $e=(e^1,\ldots,e^k)\in \mathcal{S}$ represents the set of items on each stack that have already been picked up. Therefore, $e^\ell = h$ means that items $v^\ell_1,\ldots,v^\ell_{h}$ have been collected, and that the current bottom (that is, its $(h+1)$th element $v^\ell_{h+1}$) is the next item in $P^\ell$ that has to be picked up. Let denote with $W(e)=\cup_{\ell=1}^k\{v^\ell_1,\ldots,v^\ell_{e^\ell}\}$ the set of collected items once the state $e$ has been reached. Althought there are (in general) an exponential number of paths to reach a given state $e$, there are only (at most) $k$ possible preceding states, depending on which stack has been considered at last. Hence, we associate to each state $e$ its list of possible predecessors $p(e,1),\ldots,p(e,k)$, where $p(e,\ell)=(e^1,\ldots,e^{\ell-1},e^\ell-1,e^{\ell+1},\ldots,e^k)$, for $e$ and $\ell$ such that $e^\ell \geq 1$. 

In order to build an optimal tour, we associate to each state $e$ a collection of $k$ labels that correspond to its cost. For any $e\neq(0,\ldots,0)\in\mathcal{S}$, and for any $\ell\in[k]$, the label $\mathcal{E}(e,\ell)$ gives the minimum cost for picking up all the items of $W(e)$ starting from 0, compatible with $\mathcal{P}$, and that end with $P^\ell$. According to this definition,  $\mathcal{E}(e,\ell)$ may only depend on $\mathcal{E}(p(e,\ell),\ell')$, for $\ell'\in[k]$: the current sub-tour $T^1$ may reach $e$ after having picked up $v^\ell_{e^\ell}$ iff $v^\ell_{e^\ell}$ has not been picked up yet. Then, for  any $e\neq(0,\ldots,0),(q^1,\ldots,q^k)\in \times_{\ell=1}^k[0,q^\ell]$ and for any $\ell\in[k]$, we have the following reccurrence relation:
$$\mathcal{E}(e,\ell)= \left\{\begin{array}{ll}
	+\infty										&\mbox{if }e^\ell = 0\\
	\min_{\ell'=1}^k \{\mathcal{E}(p(e,\ell),\ell')+d^1(v^{\ell'}_{p(e,\ell)^{\ell'}},v^{\ell}_{e^{\ell}}) \ |\ p(e,\ell)^{\ell'} \geq 1\}
														&\mbox{if }e^\ell \geq 1
\end{array}\right.$$

Note that item $v^{\ell'}_{p(e,\ell)^{\ell'}}$ differs from $v^{\ell'}_{e^{\ell'}}$ (that is, $p(e,\ell)^{\ell'}$ differs from $e^{\ell'}$) iff $\ell'=\ell$. The initial conditions are given by the $k$ states $f(\ell)=(0,\ldots,0,1,0,\ldots,0)$ that correspond to the $\ell$th canonical vectors: 
$$\mathcal{E}(f(\ell),\ell')= \left\{\begin{array}{ll}
	+\infty\mbox{ if }\ell'\neq \ell, &d^1(0, v^{\ell}_{1})\mbox{ if }\ell'=\ell
\end{array}\right\}$$

Finally, the expression of the labels on the final state $F=(q^1,\ldots,q^k)$ is the following (for $\ell$ such that $F^\ell\geq 1$):
$$\mathcal{E}(F,\ell)= 
	\min_{\ell'=1}^k \{\mathcal{E}(p(F,\ell),\ell')+d^1(v^{\ell'}_{p(F,\ell)^{\ell'}},v^{\ell}_{q^{\ell}})+d^1(v^{\ell}_{q^{\ell}},0) \ |\ p(F,\ell)^{\ell'} \geq 1\}$$

\begin{algorithm}\label{algo-empiler}
\caption{optimal PICKUP TOUR($\mathcal{P}$)}
\KwIn{$I=(d^1,d^2)$ an instance of $k$STSP, $\mathcal{P}$ a stacking order on $I$.}
\KwOut{An optimal pickup tour $T^1$ that is compatible with $\mathcal{P}$.}
\BlankLine

\tcp{Initialization stage}
{\bf for} $e\in\mathcal{S}$, $\ell\in[k]$ s.t. $e^\ell = 0$ {\bf do} $\mathcal{E}(e,\ell) \longleftarrow +\infty$\;
\BlankLine

{\bf for} $\ell\in[k]$ {\bf do} $\mathcal{E}(f(\ell), \ell) \longleftarrow d^1(0,v^\ell_1)$\;
\BlankLine

\tcp{Dynamic procedure}
\For{$p=2$ to $n-1$}
{
	\For{$e\in\mathcal{S}$ s.t. $|e|=p$}
	{
		\For{$\ell= 1$ to $k$ s.t. $e^\ell \geq 1$}
		{
			$\mathcal{E}(e,\ell) \longleftarrow 
					\min_{\ell'=1}^k\{\ \mathcal{E}(p(e,\ell),\ell') + d^1(v^{\ell'}_{p(e,\ell)^{\ell'}},v^\ell_{e^{\ell}})\ |\ p(e,\ell)^{\ell'}\geq 1\ \}$\;
		}
	}
}
\BlankLine

\tcp{Termination}
\For{$\ell= 1$ to $k$ s.t. $F^\ell \geq 1$}
{
	$\mathcal{E}(F,\ell) \longleftarrow 
			\min\limits_{\ell'=1}^k\{\ \mathcal{E}(p(F,\ell),\ell')+d^1(v^{\ell'}_{p(F,\ell)^{\ell'}},v^\ell_{F^{\ell}})+d^1(v^\ell_{F^{\ell}},0)\ |\ p(F,\ell)^{\ell'}\geq 1\ \}$\;
}
\BlankLine
		
\Return $T^1$ the tour associated with the label $\arg\min\{\mathcal{E}(F,q^\ell)\ |\ \ell=1,\ldots,k\}$\;
\end{algorithm}

The optimal value is the minimal cost among $\mathcal{E}(F,1),\ldots,\mathcal{E}(F,k)$, since any feasible pickup tour must end with the top of some stack. Furthermore, the recurrence relation indicates that the labels of a state $e$ (including $F$) such that $|e|=p$ (where $|\cdot|$ denotes the Hamming norm) only depend on a subset of the states $e'$ such that $|e|'=p-1$. Based on these observations, Algorithm~\ref{algo-empiler} computes an optimal pickup tour within polynomial time. The number of states to consider is upper bounded by $(n+1)^k-1$ (the worst configuration occurs when the items are fairly distributed in the $k$ stacks). The computation of the $k$ labels of a given state requires at worst $k^2$ comparisons and the global complexity is $\mathcal{O}((n+1)^{k})$. Note that the computation of an optimal delivery tour is perfectly symetric, considering the reverse order on each stack.
\qed\end{proof}

\section{Evaluating optimal $k$STSP {\em vs.} optimal TSP}\label{sec-apx}
In this section, we discuss relationships between solution values of the two TSP tours and the optimal value of the $k$STSP. For a given instance $I=(d^1,d^2)$ of the $k$STSP, we define the two instances $I^1=(K_{n+1},d^1)$ and $I^2=(K_{n+1},d^2)$ of the TSP. The optimal values ({\em resp.}, the worst solution values) on $I^1$, $I^2$ (for the TSP) and $I$ (for the $k$STSP) are respectively denoted by $\mathrm{opt}_{TSP}(I^1)$, $\mathrm{opt}_{TSP}(I^2)$ and $\mathrm{opt}_{kSTSP}(I)$ ({\em resp.}, by $\mathrm{wor}_{TSP}(I^1)$, $\mathrm{wor}_{TSP}(I^2)$ and $\mathrm{wor}_{kSTSP}(I)$). For any $I$, these extremal values obviously verify relations~(\ref{eq-optTSP-exact-bad-apx-1}) and~(\ref{eq-optTSP-exact-bad-apx-2}). Any feasible couple $(T^1,T^2)$ for the $k$STSP is feasible for the couple of TSP instances $(I^1,I^2)$. For any tour $T^1$ on $I^1$ ({\em resp.}, $T^2$ on $I^2$), there exists a compatible tour $T^2$ ({\em resp.}, $T^1$). Note that for this latter fact (and thus, for relation~\ref{eq-optTSP-exact-bad-apx-3}), we must assume that the underlying graphs are complete.
	\begin{eqnarray}
	\label{eq-optTSP-exact-bad-apx-1}
	\mathrm{opt}_{TSP}(I^1)+\mathrm{opt}_{TSP}(I^2)\leq\mathrm{opt}_{kSTSP}(I)\\
	\label{eq-optTSP-exact-bad-apx-2}
	\mathrm{wor}_{kSTSP}(I)\leq \mathrm{wor}_{TSP}(I^1)+\mathrm{wor}_{TSP}(I^2)\\
	\label{eq-optTSP-exact-bad-apx-3}
	\mathrm{opt}_{kSTSP}(I)\leq \left\{\begin{array}{c}
																		\mathrm{opt}_{TSP}(I^1)+\mathrm{wor}_{TSP}(I^2)\\
																		\mathrm{wor}_{TSP}(I^1)+\mathrm{opt}_{TSP}(I^2)
																		\end{array}\right\} \leq\mathrm{wor}_{kTSP}(I)
	\end{eqnarray}
\noindent In particular we discuss the results obtained by two simple heuristics denoted here after TWS and TWD and based on the idea of solving to optimality a single TSP. TWS builds a solution of the kSTSP, by solving to optimality the delivery tour, while the pickup tour is fixed (or the reverse). TWD determines a solution for the kSTSP, by solving to optimality a single stack TSP on a graph whose distance function is obtained by summing up the original distance functions. We prove that both TWS and TWD give an unbouded error, when particular instances families are considered. Let's introduce some more notations. Given the TSP solutions $T'^1,T'^2$ for $I^1,I^2$, we denote by $\mathrm{opt}_{2STSP|T'^{1}}$ ({\em resp.} $\mathrm{opt}_{2STSP|T'^{2}}$) the best solution value for the $k$STSP on $I$, among the solutions $(T^1,T^2,\mathcal{P})$ where $T^1=T'^1$ ({\em resp.}, $T^2=T'^2$). The optimal TSP tours on $I^1,I^2$ are denoted by $T^{1,*},T^{2,*}$, respectively. Moreover, for a given $\alpha\in]0,1[$, we denote by $I_\alpha$ the TSP instance on $K_{n+1}$ with distance function $d_\alpha=2\left(\alpha d^1 + (1-\alpha)(d^2)^{-1}\right)$, where $(d^2)^{-1}$ is defined as $(d^2)^{-1}(a,b)=d^2(a,b)$ for any couple $(a,b)$ of items. The optimal tour on $I_\alpha$ will be denoted by $T^*_\alpha$.

\begin{lemma}\label{lem-optTSP-exact-bad-apx}
We consider arbitrary distance functions $d^1,d^2$ (symetric or not).
\begin{enumerate}
\item\label{lem-optTSP-exact-bad-apx-1}For $a\in\{1,2\}$, the optimal value $\mathrm{opt}_{2STSP|T^{*,a}}(I)$ verifies: 
$$\begin{array}{cccc}
	\inf_{I\in I_{2STSP}}	&\mathrm{opt}_{2STSP|T^{a,*}}(I)/\mathrm{opt}_{2STSP}(I)																					&=	&+\infty
\end{array}$$
\item\label{lem-optTSP-exact-bad-apx-2}For $\alpha=1/2$, the quantities $d^1(T^*_\alpha)+d^2(T^*_\alpha)$ and $\mathrm{opt}_{2STSP}(I)$ verify: 
	$$\inf_{I\in I_{2STSP}}\left(d^1(T^*_\alpha)+d^2(T^*_\alpha)\right)/\mathrm{opt}_{2STSP}(I)\ =\ +\infty$$	
\end{enumerate}
\end{lemma}

\begin{proof}
\ref{lem-optTSP-exact-bad-apx}-\ref{lem-optTSP-exact-bad-apx-1}, the asymetric case. Consider the instance family $\left(I_n\right)_{n\geq 3}$, $I_n=(d^1_n,d^2_n)$, defined as (indexes are taken modulo $n+1$):
$$\begin{array}{ll}
	d^1_n(u,v)=\	\left\{\begin{array}{lr}
						1				&\mbox{if }v=u+1\\
						1+\varepsilon	&\mbox{otherwise}
				\end{array}\right.	&d^2_n(u,v)=\	\left\{\begin{array}{ll}
																									1						&\mbox{if }v=u+1\\
																									n						&\mbox{otherwise}
																								\end{array}\right.
\end{array}$$

The optimal values for TSP on $I^1_n$ and $I^2_n$ both are $n+1$, reached by the tours $T_n^{1,*}=T_n^{2,*}=(0,1,2,\ldots,n,0)$. If we fix $T_n'^{1}=T_n^{1,*}$, then any stacking order $\mathcal{P}_n=\{P_n^1,P_n^2\}$ that is compatible with $T_n'^{1}$ will order the items in such a way that contradicts the order induced by the optimal delivery tour. In order to evaluate the best possible compatible $T_n^2$, we consider three cases:
\begin{itemize}
\item Case $(0,1)\in T_n^2$: $1>^2 u\ \forall u\neq 1\in[n]$, $1<^1 u\ \forall u\neq 1$ and thus, item $1$ is non comparable to any other item under $<^3$; hence, $P_n^a=(1), P_n^{3-a}=(2,\ldots,n)$ (for $a\in\{1,2\})$, $T_n^2=(0,1,n,n-1,\ldots,3,2,0)$ and $d_n^2(T_n^2)= 1+n^2$. 
\item Case $(n,0)\in T_n^2$: by means of a similar argument, we obtain $P_n^a=(n), P_n^{3-a}=(1,\ldots,n-1)$, $T_n^2=(0,n-1,n-2,\ldots,2,1,n,0)$ and $d_n^2(T_n^2)= 1+n^2$. 
\item Case $(0,1),(n,0)\notin T_n^2$: consider any item $u\in [n]\backslash \{1,n\}$, and assume $(u-1,u),(u,u+1)\in T_n^2$; we then deduce from $u-1>^2 u >^2 u+1$ and $u-1<^1 u <^1 u+1$ that items $u-1$, $u$ and $u+1$ are pairwise non comparable under $<^3$, what is not possible for $k=2$. Hence, $T_n^2$ uses at least one arc ({\em resp.}, exactly 2 arcs) of distance~$n$ per item $v\in[n]$ ({\em resp.}, for the depot $0$) and thus, $d_n^2(T_n^2)\geq 1/2(n(n+1)+2n)=n(n+3)/2$. 
\end{itemize}

Moreover, the following solution is optimal for $I_n$, of value $2(n+1)+(\lfloor(n+1)/2\rfloor+1)\varepsilon$; it consists in stacking items of odd and even values by decreasing order into separated containers, which enables to use the delivery tour $T_n^2=T^{2,*}_n$, while putting into $T_n^1$ a maximum number of arcs of kind $(u,u+1)$ (figure~\ref{fig-optTSP-exact-bad-apx-1} illustrates instance $I_n$ for even value of $n$):
$$\begin{array}{lcl}
	P_n^1			&=	&\left(n, n-2, \ldots, n-2i, \ldots, n-2\lfloor (n-1)/2\rfloor\right)\\
	P_n^2			&=	&\left(n-1, n-3, \ldots, n-(2i+1), \ldots, n-(2\lfloor (n-2)/2\rfloor+1)\right)\\
	\mathcal{T}_n 	&= 	&\{(0,n-1)\}\cup\{(n-(2i+1), n-2i,n-(2i+3))\ |\ i=0,\ldots,\lfloor (n-4)/2\rfloor\}\\
	T_n^1 			&= 	&\mathcal{T}\cup\{(2,0)\}\mbox{ if }n\mbox{ even},\ 
						\mathcal{T}\cup\{(3,1),(1,0)\}\mbox{ if }n\mbox{ odd}\\
	T_n^2			&=	&(0,1,\ldots,n,0)
\end{array}$$

We thus get the expected result (note by the way that, since $\mathrm{opt}_{TSP}(I_n^1)+\mathrm{wor}_{TSP}(I_n^2)=(n+1)^2$, the ratio $\mathrm{opt}_{2STSP|T_n'^{1}}(I_n)/(\mathrm{opt}_{TSP}(I_n^1)+\mathrm{wor}_{TSP}(I_n^2))$ is asymptotically $1/2$):
$$\left.\begin{array}{lcl}
	\mathrm{opt}_{2STSP|T_n'^{1}}(I_n)	&\geq &(n+1)+n(n+3)/2\\
	\mathrm{opt}_{2STSP}(I_n)					&\leq &(n+1)(2+\varepsilon)
	\end{array}\right\}\ \Rightarrow\ 
	\frac{\mathrm{opt}_{2STSP|T_n'^{1}}(I_n)}{\mathrm{opt}_{2STSP}(I_n)}\mathop{\longrightarrow}_{n\rightarrow +\infty} +\infty$$

\ref{lem-optTSP-exact-bad-apx}-\ref{lem-optTSP-exact-bad-apx-1}, the symetric case. 
Consider $\left(J_{n}\right)_{{n}\geq 6}$, $J_{n}=(d^1_{n},d^2_{n})$, defined as (see figure~\ref{fig-optTSP-exact-bad-apx-2} for an illustration):
$$\begin{array}{ll}
	d^1_{n}(u,v)=\	\left\{\begin{array}{lr}
						1							&\mbox{if }v=u\pm1\\
						1+\varepsilon	&\mbox{otherwise}
				\end{array}\right.	&d^2_{n}(u,v)=\	\left\{\begin{array}{ll}
																								1	&\mbox{if }v+u\in\{n,n+1\}\\
																								n	&\mbox{otherwise}
																						\end{array}\right.
\end{array}$$
The tours $T_n^{1,*}=(0,1,2,\ldots,n,0)$ and $T_n^{2,*}=(0,n,1,n-1,2,n-2,3\ldots,\lceil n/2\rceil,0)$ are optimal on $J_n^1,J_n^2$, of value $n+1$ and $2n$, respectively. Similary to the asymetric case, we show that, for items $u=2,\ldots,n-1$ such that $2u\notin\{n-1,n,n+1,n+2\}$, any tour $T_n^2$ thas is compatible with $T_n'^1=T_n^{1,*}$ cannot use the whole (undirected) sequence\footnote{due to lack of space, the proof is not provided here; please contact the authors for further information}:
	$$\{u+1,n-u,u,n+1-u,u-1\}$$

Hence, $\mathrm{opt}_{2STSP|T_n'^1}(J_n)\geq (n+1)+((n-4)(3+n)+5*4)/4\geq n^2/4$, whereas the following solution is of value $3n+2\varepsilon-1$ (case $n$ even):
$$\begin{array}{l}
	P_n^1=(1,2,3,\ldots,n/2),\ \ \ P_n^2=(n, n-1,n-2,\ldots,n/2+1)\\
	T_n^1=(0,1,2,\ldots,n/2;n,n-1,\ldots,n/2+2,n/2+1;0)\\
	T_n^2=(0,n,1,n-1,2,\ldots,n/2-1,n/2+1,n/2;0)
\end{array}$$

\ref{lem-optTSP-exact-bad-apx}-\ref{lem-optTSP-exact-bad-apx-2}. Consider the following symetric instance family $\left(H_{n}\right)_{{n}\geq 3}$:
$$\begin{array}{|ll|ll|l|}\hline
	\mbox{condition}				&										&d^1_{n}(u,v)		&d^2_{n}(u,v)	&d^1_{n}(u,v)+d^2_{n}(u,v)\\\hline
	\mbox{{\bf if }}				&v=u\pm 1						&1							&n						&n+1\\
	\mbox{{\bf else if }}		&v=u\pm 2						&1							&n+1					&n+2\\
	\mbox{{\bf else if }}		&u+v\in\{n+1,n+3\}	&n+1						&1						&n+2\\
	\mbox{{\bf else}}				&										&n+1						&n+1					&2n+2\\\hline
\end{array}$$

The tour $T^*_{n,1/2}=(0,1,2,\ldots,n,0)$ that is optimal for the aggregate distance function is of value $(n+1)^2$, whereas there exists a couple $(T_n^{1,*},T_n^{2,*})$ of compatible optimal tours with $d_n^1(T_n^{1,*})=n+1$ and $d_n^2(T_n^{2,*})=(n-3)+5(n+1)$ (for $n$ even) or $d_n^2(T_n^{2,*})=(n-4)+6(n+1)$ (for $n$ odd); hence, the following solution of 2STSP is optimal, of value in $\{7n+2,8n+2\}$ (case $n$ even):
$$\begin{array}{l}
	P_n^1=(1,3,5,\ldots,n-3,n-1),\ \ \ P_n^2=(n, n-2,n-4,\ldots,4,2)\\
	T_n^1=(0;1,3,\ldots,n-3,n-1;n,n-2,\ldots,4,2,0)\\
	T_n^2=(0;1,n,3,n-2,5,\ldots,6,n-3,4,n-1,2;0)
\end{array}$$

Note that there exist simplier instance families verifying that $d_{n,1/2}(T^{*}_{n,1/2})$ is arbitrarly large {\em vs.} $opt_{2STSP}$; however, for more relevancy, we built $(H_n)_{(n)}$ in such a way that $T^{*}_{n,1/2}$ and the couple $(T_n^{1,*},T_n^{2,*})$ are not compatible.
\qed\end{proof}

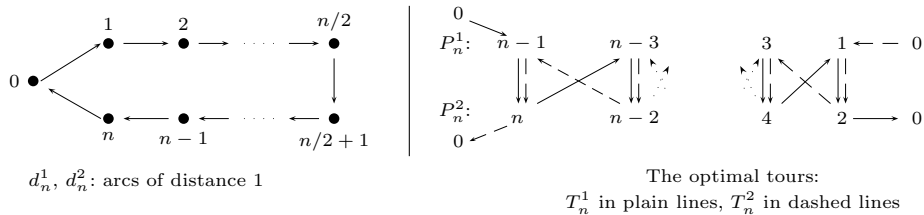
\begin{figure}[t]
\begin{center}
\begin{pspicture}(-0.3,-0.8)(14.7,1.3)
\qdisk(0,0.5){2pt}\uput[180](0,0.5){\scriptsize{0}}
\pcline[linewidth=0.2pt]{->}(0.1,0.5)(0.9,1)
\pcline[linewidth=0.2pt]{->}(0.9,0.0)(0.2,0.4)
\pcline[linewidth=0.2pt]{->}(4,0.8)(4,0.2)
\qdisk(1,1){2pt}\uput[90](1,1){\scriptsize{1}}
\pcline[linewidth=0.2pt]{->}(1.2,1)(1.8,1)
\qdisk(2,1){2pt}\uput[90](2,1){\scriptsize{2}}
\pcline[linewidth=0.2pt]{->}(2.2,1)(2.6,1)
\pcline[linestyle=dotted, linewidth=0.4pt]{-}(2.8,1)(3.2,1)
\pcline[linewidth=0.2pt]{->}(3.4,1)(3.8,1)
\qdisk(4,1){2pt}\uput[90](4,1){\scriptsize{$n/2$}}
\qdisk(1,0){2pt}\uput[-90](1,0){\scriptsize{$n$}}
\pcline[linewidth=0.2pt]{<-}(1.2,0)(1.8,0)
\qdisk(2,0){2pt}\uput[-90](2,0){\scriptsize{$n-1$}}
\pcline[linewidth=0.2pt]{<-}(2.2,0)(2.6,0)
\pcline[linestyle=dotted, linewidth=0.4pt]{-}(2.8,0)(3.2,0)
\pcline[linewidth=0.2pt]{<-}(3.4,0)(3.8,0)
\qdisk(4,0){2pt}\uput[-90](4,0){\scriptsize{$n/2+1$}}

\uput[-90](1.5,-0.5){\scriptsize{$d^1_n$, $d^2_n$: arcs of distance $1$}}
\uput[180](6,1){\scriptsize{$P_n^1$:}}
\uput[180](7,1){\scriptsize{$n-1$}}
\uput[180](8.5,1){\scriptsize{$n-3$}}
\uput[180](10,1){\scriptsize{$3$}}
\uput[180](11,1){\scriptsize{$1$}}
\uput[180](6,0.1){\scriptsize{$P_n^2$:}}
\uput[180](6.7,0){\scriptsize{$n$}}
\uput[180](8.5,0){\scriptsize{$n-2$}}
\uput[180](10,0){\scriptsize{$4$}}
\uput[180](11,0){\scriptsize{$2$}}
\uput[180](5.9,1.4){\scriptsize{0}}
\uput[180](5.9,-0.3){\scriptsize{0}}
\uput[180](12,1){\scriptsize{0}}
\uput[180](12,0){\scriptsize{0}}
\pcline[linewidth=0.2pt]{->}(5.8,1.3)(6.3,1.1)			
\pcline[linewidth=0.2pt]{->}(10.9,0)(11.5,0)				
\pcline[linewidth=0.2pt]{->}(6.45,0.8)(6.45,0.2)		
\pcline[linewidth=0.2pt]{->}(6.7,0.2)(7.8,0.8)			
\pcline[linewidth=0.2pt]{->}(7.95,0.8)(7.95,0.2)		
\pcline[linestyle=dotted, linewidth=0.4pt]{->}(8.2,0.2)(8.45,0.7)		
\pcline[linestyle=dotted, linewidth=0.4pt]{->}(9.4,0.3)(9.65,0.8)		
\pcline[linewidth=0.2pt]{->}(9.7,0.8)(9.7,0.2)			
\pcline[linewidth=0.2pt]{->}(9.95,0.2)(10.6,0.8)		
\pcline[linewidth=0.2pt]{->}(10.7,0.8)(10.7,0.2)		
\pcline[linestyle=dashed, linewidth=0.2pt]{<-}(10.9,1)(11.5,1)				
\pcline[linestyle=dashed, linewidth=0.1pt]{<-}(5.8,-0.3)(6.3,-0.1)		
\pcline[linestyle=dashed,linewidth=0.2pt]{->}(6.55,0.8)(6.55,0.2)			
\pcline[linestyle=dashed,linewidth=0.2pt]{->}(8.05,0.8)(8.05,0.2)			
\pcline[linestyle=dashed,linewidth=0.2pt]{->}(7.8,0.2)(6.7,0.8)				

\pcline[linestyle=dashed,linestyle=dotted, linewidth=0.4pt]{->}(9.65,0.2)(9.4,0.7)		
\pcline[linestyle=dashed,linestyle=dotted, linewidth=0.4pt]{->}(8.45,0.3)(8.2,0.8)		

\pcline[linestyle=dashed,linewidth=0.2pt]{->}(9.8,0.8)(9.8,0.2)			
\pcline[linestyle=dashed,linewidth=0.2pt]{->}(10.6,0.2)(9.9,0.8)		
\pcline[linestyle=dashed,linewidth=0.2pt]{->}(10.8,0.8)(10.8,0.2)		

\uput[-90](9.3,-0.5){\scriptsize{The optimal tours:}}
\uput[-90](9.3,-0.8){\scriptsize{$T_n^1$ in plain lines, $T_n^2$ in dashed lines}}
\pcline[linewidth=0.4pt]{-}(5,-0.5)(5,1.5)
\end{pspicture}
\caption{\label{fig-optTSP-exact-bad-apx-1}\small{Instance $I_n$ for $n$ even.}}
\end{center}
\end{figure}

\begin{figure}[t]
\begin{center}
\begin{pspicture}(-0.3,-0.8)(14.7,1.3)
\qdisk(0,0.5){2pt}\uput[180](0,0.5){\scriptsize{0}}
\pcline[linewidth=0.2pt]{-}(0.1,0.5)(0.9,1)													
\pcline[linewidth=0.2pt]{-}(0.9,0.05)(0.15,0.40)										
\pcline[linestyle=dashed,linewidth=0.2pt]{-}(0.9,-0.05)(0.15,0.30)	
\pcline[linewidth=0.2pt]{-}(3.95,0.8)(3.95,0.2)											
\pcline[linestyle=dashed,linewidth=0.2pt]{-}(4.05,0.8)(4.05,0.2)		
\pcline[linestyle=dashed,linewidth=0.2pt]{-}(1,0.8)(1,0.2)					
\pcline[linestyle=dashed,linewidth=0.2pt]{-}(2,0.8)(2,0.2)					
\pcline[linestyle=dotted,linewidth=0.4pt]{-}(2.2,0.9)(2.4,0.5)			
\pcline[linestyle=dotted,linewidth=0.4pt]{-}(3.7,0.5)(3.8,0.1)			
\pcline[linestyle=dashed,linewidth=0.2pt]{-}(1.1,0.9)(1.9,0.1)			
\qdisk(1,1){2pt}\uput[90](1,1){\scriptsize{1}}
\pcline[linewidth=0.2pt]{-}(1.2,1)(1.8,1)														
\qdisk(2,1){2pt}\uput[90](2,1){\scriptsize{2}}
\pcline[linewidth=0.2pt]{-}(2.2,1)(2.6,1)														
\pcline[linestyle=dotted, linewidth=0.4pt]{-}(2.8,1)(3.2,1)					
\pcline[linewidth=0.2pt]{-}(3.4,1)(3.8,1)														
\qdisk(4,1){2pt}\uput[90](4,1){\scriptsize{$n/2$}}
\qdisk(1,0){2pt}\uput[-90](1,0){\scriptsize{$n$}}
\pcline[linewidth=0.2pt]{-}(1.2,0)(1.8,0)														
\qdisk(2,0){2pt}\uput[-90](2,0){\scriptsize{$n-1$}}
\pcline[linewidth=0.2pt]{-}(2.2,0)(2.6,0)														
\pcline[linestyle=dotted, linewidth=0.4pt]{-}(2.8,0)(3.2,0)					
\pcline[linewidth=0.2pt]{-}(3.4,0)(3.8,0)														
\qdisk(4,0){2pt}\uput[-90](4,0){\scriptsize{$n/2+1$}}
\uput[-90](1.9,-0.5){\scriptsize{$d^1_n$ (plain lines), $d^2_n$ (dashed lines):}}
\uput[-90](1.9,-0.9){\scriptsize{edges of distance $1$}}
\uput[180](6,1){\scriptsize{$P^1$:}}
\uput[180](6.7,1){\scriptsize{$1$}}
\uput[180](8.2,1){\scriptsize{$2$}}
\uput[180](10.3,1.2){\scriptsize{$n/2-1$}}
\uput[180](11.2,1.2){\scriptsize{$n/2$}}
\uput[180](6,0.1){\scriptsize{$P^2$:}}
\uput[180](6.7,0){\scriptsize{$n$}}
\uput[180](8.5,-0.2){\scriptsize{$n-1$}}
\uput[180](10.3,-0.2){\scriptsize{$n/2+2$}}
\uput[180](11.5,-0.2){\scriptsize{$n/2+1$}}
\uput[180](5.9,1.4){\scriptsize{0}}
\uput[180](5.9,-0.3){\scriptsize{0}}
\uput[180](12,1){\scriptsize{0}}
\uput[180](12,0){\scriptsize{0}}
\pcline[linewidth=0.2pt]{->}(5.8,1.3)(6.3,1.1)			
\pcline[linewidth=0.2pt]{->}(7,1)(7.8,1)						
\pcline[linestyle=dotted, linewidth=0.4pt]{->}(8.2,1.0)(8.8,1.0)		
\pcline[linestyle=dotted, linewidth=0.4pt]{->}(9.1,1.0)(9.7,1.0)		
\pcline[linewidth=0.2pt]{->}(10.1,1)(10.7,1)				
\pcline[linewidth=0.2pt]{->}(10.9,0)(11.5,0)				
\pcline[linewidth=0.2pt]{->}(7,0)(7.8,0)						
\pcline[linestyle=dotted, linewidth=0.4pt]{->}(8.2,0.0)(8.8,0.0)		
\pcline[linestyle=dotted, linewidth=0.4pt]{->}(9.1,0.0)(9.7,0.0)		
\pcline[linewidth=0.2pt]{->}(10.1,0)(10.7,0)				
\pcline[linewidth=0.2pt]{->}(10.7,0.9)(6.8,0.1)		
\pcline[linestyle=dashed, linewidth=0.2pt]{<-}(10.9,1)(11.5,1)				
\pcline[linestyle=dashed, linewidth=0.1pt]{<-}(5.8,-0.3)(6.3,-0.1)		
\pcline[linestyle=dashed,linewidth=0.2pt]{->}(6.5,0.8)(6.5,0.2)				
\pcline[linestyle=dashed,linewidth=0.2pt]{->}(8,0.8)(8,0.2)						
\pcline[linestyle=dashed,linewidth=0.2pt]{->}(7.8,0.2)(6.7,0.8)				
\pcline[linestyle=dashed,linestyle=dotted, linewidth=0.4pt]{->}(9.65,0.1)(9.4,0.6	)		
\pcline[linestyle=dashed,linestyle=dotted, linewidth=0.4pt]{->}(8.45,0.4)(8.2,0.9)		
\pcline[linestyle=dashed,linewidth=0.2pt]{->}(9.8,0.8)(9.8,0.2)			
\pcline[linestyle=dashed,linewidth=0.2pt]{->}(10.6,0.2)(9.9,0.8)		
\pcline[linestyle=dashed,linewidth=0.2pt]{->}(10.8,0.8)(10.8,0.2)		
\uput[-90](9.3,-0.5){\scriptsize{The optimal tours:}}
\uput[-90](9.3,-0.8){\scriptsize{$T_n^1$ in plain lines, $T_n^2$ in dashed lines}}
\pcline[linewidth=0.4pt]{-}(5,-0.5)(5,1.5)
\end{pspicture}
\caption{\label{fig-optTSP-exact-bad-apx-2}\small{Instance $J_n$ for $n$ even.}}
\end{center}
\end{figure}
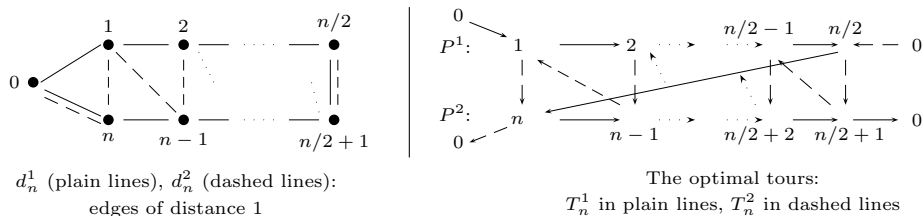

\section{Conclusion}\label{sec-conc}
This paper address the time complexity of $k$STSP, whose highly combinatorial structure suggests the search of approximation algorithms (may be the most likely for the differential ratio, \cite{M02}). The good complexity of its sub-problems makes relevant the design of exact methods based on constraints decomposition of $k$STSP. Finally, it would be interesting to better characterize the shape of precedence constraints for which TSP/sequencing under precedence contraints are tractable. Indeed, we deduce from Theorem~\ref{thm-TSP} that TSP under {\em ``stack precedence constraints''} is in $\mathbb{P}$. Equivalently, the single machine scheduling with sequence-dependent time or cost setup under the same shape of constraints, that we denote by (1/$k$-stack,$p$,$ST_{sd}$/$C_{max}$) $\equiv$ (1/$k$-stack,$p$,$ST_{sd}$/TST) and (1/$k$-stack,$p$,$SC_{sd}$/$TSC$), are tactable (for the notations used in order to represent the $\alpha$/$\beta$/$\gamma$-classification, \cite{GLLRK79} of scheduling problems, we refer to~\cite{AGA99}). Here, by {\em ``stack precedence constraints''}, we mean that the constraints define a partial order on $[n]$ within at most $k$ ordered subsets, where $k$ is a universal constant.


\end{document}